\documentclass[11pt]{article}
\usepackage[utf8]{inputenc}
\usepackage{tikz}
\usepackage[nottoc]{tocbibind}
\usepackage[colorinlistoftodos]{todonotes}
\usepackage{hyperref}
\usepackage[margin=1.2in]{geometry}
\usepackage{amsmath,amsfonts,amssymb,amsthm}
\usepackage{caption}
\usepackage{subcaption}
\usepackage{xspace}
\usepackage{multirow}
\usepackage{dsfont}
\usepackage{algpseudocode}
\usepackage{algorithm}
\usepackage{mathtools}

\renewcommand{\P}{\mathbb{P}}
\newcommand{\E}{\mathbb{E}}

\newcommand{\tmax}{t_{\textup{max}}}

\newcommand{\simple}{\textsc{Simple}\xspace}

\newcommand{\B}{\mathcal{B}}

\newcommand{\ls}{\lambda_s}
\newcommand{\wmax}{w_n(\tmax)}
\renewcommand{\tt}{\tilde{t}}
\newcommand{\ftt}{\lfloor \tilde{t} \rfloor}
\newcommand{\Sopt}{S^{\textup{opt}}}

\newtheorem{definition}{Definition}
\newtheorem{theorem}{Theorem}

\newtheorem{lemma}{Lemma}
\newtheorem{corollary}{Corollary}

\newcommand{\footremember}[2]{%
   \footnote{#2}
    \newcounter{#1}
    \setcounter{#1}{\value{footnote}}%
}
\newcommand{\footrecall}[1]{%
    \footnotemark[\value{#1}]%
} 

\begin{document}

\title{Reconstruction of geometric random graphs with the \simple algorithm}

\author{Clara Stegehuis\footremember{ut}{Department of Electrical Engineering, Mathematics and Computer Science, 
University of Twente, The Netherlands} $\&$ Lotte Weedage\footrecall{ut}}

\maketitle
\begin{abstract}
Graph reconstruction can efficiently detect the underlying topology of massive networks such as the Internet. Given a query oracle and a set of nodes, the goal is to obtain the edge set by performing as few queries as possible. An algorithm for graph reconstruction is the \simple algorithm \cite{mathieu2023simple}, which reconstructs bounded-degree graphs in $\tilde{O}(n^{3/2})$ queries. We extend the use of this algorithm to the class of geometric random graphs with connection radius $r \sim n^k$, with diverging average degree. We show that for this class of graphs, the query complexity is $\tilde{O}(n^{2k+1})$ when $k > 3/20$. This query complexity is up to a polylog$(n)$ term equal to the number of edges in the graph, which means that the reconstruction algorithm is almost edge-optimal. We also show that with only $n^{1+o(1)}$ queries it is already possible to reconstruct at least $75\%$ of the non-edges of a geometric random graph, in both the sparse and dense setting. Finally, we show that the number of queries is indeed of the same order as the number of edges on the basis of simulations. %without a multiplicative polylog$(n)$-term.
\end{abstract}

\section{Introduction}
Given only the set of nodes of a graph, how can one, with as little information as necessary, reconstruct all edges in this graph? This problem is called the \emph{graph reconstruction problem}: how to fully reconstruct a graph using a minimal number of queries. The exact nature of the query oracle used in these problems can vary, but in general such oracles have a node or a pair of nodes as input and the output gives some information on this node or pair of nodes.\\

A motivating example for the graph reconstruction problem is the discovery of the topology of the Internet \cite{motamedi2014survey} at different levels, such as the autonomous systems (AS) level or the router level. The discovery of these topologies is necessary for e.g. routing purposes or assessing the robustness of the network. However, as the network is huge and links might change over time, it is infeasible to execute too many time- and energy-consuming queries on the network. Thus, it is desirable to use as few queries as possible to reconstruct the network. As the internet showcases a large spatial dependency \cite{boguna2021network}, we aim to understand how the presence of an underlying network geometry can help in fast reconstruction.\\

A popular method to obtain the topology of the Internet is \texttt{ traceroute}. This method reveals the path that a packet chooses from source to destination over the internet. By performing \texttt{traceroute} queries from sufficiently many different sources and destinations, one can reconstruct the complete network topology \cite{motamedi2014survey}. In general, \texttt{traceroute} resembles a \emph{shortest-path} oracle that reveals the shortest path between a queried source and target. However, while \texttt{traceroute} might provide the full shortest path, often this information is not available due to e.g. privacy \cite{kannan2018graph} or the size of the network and the paths \cite{beerliova2006network}. In these cases, often a distance oracle, where only the length of the shortest path is given, is more realistic and reliable.\\

The \simple algorithm, introduced in \cite{mathieu2023simple}, is a randomized algorithm that can fully reconstruct general bounded-degree graphs on $n$ nodes in $\tilde{O}(n^{5/3})$\footnote{$\tilde{O}(f(n)) = O(f(n) \text{polylog}(n))$} distance queries. In this paper, we extend the use of the \textsc{Simple} algorithm, until now only valid for sparse, bounded-degree graphs with a maximum degree up to $O(\log(n))$, to the class of geometric random graphs with growing average degrees. We show how the \simple algorithm can reconstruct geometric random graphs (GRGs) on $n$ nodes with connection radius $r \sim n^k$ with a query complexity of $\tilde{O}(n^{2k+1})$ when $\frac{3}{20} \leq k < \frac{1}{2}$ and $\tilde{O}(n^{3/2-4k/3})$ for smaller $k$. Thus, for dense geometric random graphs with $k \geq \frac{3}{20}$, this algorithm results in an almost \emph{optimal} reconstruction: the query complexity for reconstruction is of the same order (multiplied by a log-term) as the number of edges. The results presented in this paper provide the first bounds on query complexity for reconstruction of dense graphs with unbounded maximum degree. Moreover, we show that for both dense and sparse geometric random graphs, the first step of the \simple algorithm - where all nodes are queried with a small set of seed nodes - already identifies $75\%$ of the non-edges of the graph. \\

The set-up of this paper is as follows. In Section \ref{sec:related} we provide an overview of the literature on graph reconstruction, followed by a discussion of the \textsc{Simple} algorithm and the random graph model we will use throughout this paper in Section \ref{sec:simple}. Then, in Section \ref{sec:lemmas} we prove our main results on the number of queries needed to reconstruct geometric graphs and prove our results.

\section{Related work}\label{sec:related}

\paragraph{Query oracles.}
An important aspect of the graph reconstruction problem is the choice of oracle. Some of these queries classify as \emph{global} query oracles: they receive as input a node and output all shortest paths \cite{beerliova2006network,sen2010covert} or all distances to other nodes \cite{erlebach2006network}. \emph{Local} query oracles, such as the shortest-path oracle \cite{reyzin2007learning} and the distance oracle \cite{kannan2018graph}, receive as input a node pair and output the nodes on the shortest path or the length of the shortest path between this node pair, respectively. In a similar manner, the \emph{kth-hop} oracle \cite{afshar2022mapping} outputs the $k$-th node on the shortest path between a given node pair, and the \emph{betweenness} oracle \cite{abrahamsen2016graph} has a third input node, and tells whether this node lies on the shortest path between the other two nodes. Lastly, the \emph{maximal independent set} oracle receives a subset of nodes as input and outputs any maximal independent set in the induced subgraph of this set of nodes \cite{konrad2024graph}.\\

Motivated by the reconstruction of Internet networks, our focus is on the local type of oracle: the distance oracle. In the rest of this paper, therefore, we focus only on distance queries. 

\paragraph{Graph reconstruction algorithms.}
A trivial upper bound reconstructing general graphs using the distance oracle is $O(n^2)$. Indeed, in the worst case examples of a complete graph or an empty graph, all pairs of nodes need to be queried \cite{reyzin2007learning}. Without knowledge of the graph structure, graph reconstruction therefore always takes $O(n^2)$ queries in the worst case. \\

Most results about graph reconstruction are therefore on special cases of graph classes that do not contain these worst-case example. Specifically, bounded-degree, connected graphs have been investigated thoroughly. The algorithm proposed in \cite{kannan2018graph} requires $\tilde{O}(n^{3/2})$ queries, given a maximum degree $\Delta$ of $O(\text{polylog}(n))$. This randomized algorithm is based on Voronoi cell decomposition: nodes in the graph are grouped into equal-sized clusters that then can be reconstructed independently. In \cite{mathieu2023simple} the authors propose a simpler algorithm that uses in expectation $\tilde{O}(n)$ queries in random $\Delta$-regular graphs and $\tilde{O}(n^{5/3})$ queries in general bounded degree graphs. In case of bounded degree chordal graphs, there exist reconstruction algorithms that have a query complexity of $O(\Delta^2 n\log^2(n))$\cite{rong2021reconstruction}. For the class of $k-$chordal graphs - graphs without cycles of length at least $k+1$ - an exact reconstruction algorithm is given in \cite{bastide2023optimal} with a query complexity of $O(n\log(n))$. \\

A common element of the above-mentioned graphs is that the considered graphs have bounded degree, often contain only a few cycles, and have no underlying geometry. However, real-world networks are often defined based on their underlying geometry, such as Internet networks on the \emph{Point of Presence} (PoP) level or router level \cite{motamedi2014survey}.\\

Lastly, graph reconstruction often goes hand-in-hand with graph verification \cite{kannan2018graph}. In reconstruction, the goal is to find all edges and non-edges of a graph $\mathcal{G}$. For verification one obtains a graph $\hat{\mathcal{G}}$ in the same class of graphs as $\mathcal{G}$, and one has to verify that it is equal to $\mathcal{G}$. Graph verification is important for example in error detection and detecting anomalous network behaviour \cite{castro2004network}. A trivial upper bound on the verification complexity is the reconstruction complexity, as one can always reconstruct the entire graph $\mathcal{G}$ to verify whether the two graphs are the same. For some specific graphs however, graph verification has a lower query complexity than graph reconstruction. For graphs with a maximum degree $\Delta \leq n^{o(1)}$, the greedy algorithm in \cite{kannan2018graph} verifies graphs using the distance oracle in $O(n^{1+o(1)})$, while the reconstruction algorithm has a query complexity of $\tilde{O}(n^{3/2})$. Similarly, verification of chordal graphs with maximum degree $\Delta$ can be done in $O(\delta^2 n \log(n))$ distance queries compared to $O(\delta^2 n \log^2(n))$ for reconstruction \cite{rong2021reconstruction}.

\paragraph{Graph embedding.}
Related to network reconstruction is the retrieval of the locations of the nodes in a network with underlying geometry, given only the edge list. Finding an embedding of a graph can be seen as the opposite of graph reconstruction as edges are given and the objective is to find locations of the nodes. In \cite{diaz2020learning,dani2022improved,dani2023reconstruction} the authors propose algorithms to retrieve the geometric embedding of the nodes in a geometric random graph (without knowing the original embedding) up to a displacement of $O(\sqrt{\log(n)})$ for dense graphs. Contrary to the graph reconstruction problem, here the adjacency matrix is known, while the locations of the nodes and the connection radius $r$ are unknown. The localization algorithm proposed in \cite{diaz2020learning} is based on computing an estimate for $r$, which is then used to approximate the Euclidean distance based on the graph distance. However, while these algorithms give a precise location of the node, they all require knowledge of the edge list of the graph. Therefore, such a localization algorithm cannot be used for graph reconstruction. 

\section{Model and main results}\label{sec:simple}
In this paper, we consider a geometric random graph $\mathcal{G}(n,r) = (V, E)$ consisting of $n = |V|$ nodes. Nodes that are located within radius $r$ of each other connect, considering a Euclidean distance measure. The locations of the nodes $v\in V$ are distributed by a Poisson point process (PPP) $\Phi$ with intensity 1 on a two-dimensional $[0, \sqrt{n}] \times [0, \sqrt{n}]$ square with torus boundary conditions. To simplify notation, we identify the node $u \in V$ and its position $(x_u, y_u) \in \Phi$ both by $u$ and will use the words node and point interchangeably. We denote the graph or shortest path distance between node $u$ and $v$ with $d_G(u, v)$, and the Euclidean distance between the two nodes with $d_E(u,v)$ . \\

We define a thinning $\Phi_S$ on $\Phi$ as a randomly-chosen sample of the nodes in $\Phi$, resulting in a new PPP with intensity $\lambda_s$. Due to the nature of the PPP, the number of nodes in any area $A$, defined as $N(A)$ or $N_s(A)$, respectively, is Poisson distributed.\\

The \textsc{Simple} algorithm is an algorithm designed to recover connected general bounded degree graphs \cite{mathieu2023simple} (Algorithm \ref{alg:simple}). As a query oracle, the algorithm applies \emph{distance} queries. For every queried pair, the oracle returns the number of edges on the shortest path between node $u$ and $v$. As a result, the algorithm returns the full edge set of the graph and has query complexity $\tilde{O}(n^{5/3})$. Although faster algorithms exist for general graph reconstruction (e.g. \cite{kannan2018graph}), we specifically analyze this algorithm as it is easy to follow and we can show that it performs asymptotically optimal on GRGs.\\

\begin{algorithm}
  \caption{\textsc{Simple}$(V,s)$, from \cite{mathieu2023simple}. The function \textsc{Query}$(u,v)$ returns the number of edges on the shortest path between node $u$ and $v$. }\label{alg:simple}
  \begin{algorithmic}[1]
    \State $S \gets$ sample of $s$ nodes uniformly and independently at random from $V$
    \For{$u \in S$ and $v \in V$}
        \State \textsc{Query}$(u,v)$
    \EndFor
    \State $\hat{E} \gets$ set of node pairs $\{a, b\} \subseteq V$ such that, for all $u \in S$, $|d_G(u,a) - d_G(u,b)|\leq 1$ \label{line:condition}
    \For{$\{a,b\}\in \hat{E}$}\label{line:second}
        \State \textsc{Query}$(a,b)$
    \EndFor 
    \State \Return set of node pairs $\{a,b\} \in \hat{E}$ such that $d_G(a,b) = 1$
  \end{algorithmic}
  \end{algorithm}

The \simple algorithm consists of two steps. First it calculates the shortest path distance between the seeds and every other node in the graph. Based on these queries, Line \ref{line:condition} decides whether a node is \emph{distinguishable} or not, where distinguishable is defined as follows \cite{mathieu2023simple}:
\begin{definition}[distinguishable]\label{def:distinguishing}
    We call a node pair $u, v \in V$ \textit{distinguishable} if there is a seed $s \in S$ such that the condition in Algorithm \ref{alg:simple} line \ref{line:condition} is not satisfied, i.e. 
    \begin{align}
        |d_G(s, u) - d_G(s,v)| > 1. \label{eq:distinguishable}
    \end{align}
    If \eqref{eq:distinguishable} is satisfied for a seed $s \in S$, we say this seed distinguishes the node pair $u,v$. The seed $s$ is then called a distinguisher for node pair $u, v$.
\end{definition} 
When a pair of nodes $\{a, b\}$ is distinguishable, this means that the shortest path distances from seed $s$ to $a$ and from $s$ to $b$ differ by more than one. In that case, there cannot be an edge between node $a$ and $b$. Suppose $d_G(s, a) < d_G(s,b)$: if there would have been an edge between node $a$ and $b$, this edge would have been in the shortest path from $s$ to $b$ and the difference between the two graph distances would be 1. Then, in the second step, all pairs of nodes that are not distinguishable are queried again to check whether there is an edge between them or not.\\

The proof that \textsc{Simple} fully reconstructs random regular graphs is heavily based on the fact that the graph is tree-like, which is not the case for GRGs. Thus, to analyze the query complexity of the \simple algorithm for GRGs, we use results from \cite{dani2023reconstruction} to leverage the connection between the graph distance and the Euclidean distance to reconstruct the graph.

\begin{theorem}[Theorem 3.9 from \cite{dani2023reconstruction}]\label{thm:lowerbound_upperbound}
    There exist absolute constants $C_1, C_2$ and $c$ such that, for all $n \geq 1$, all $r \geq C_1\sqrt{\log n}$, with probability at least $1 - C_2/n^2$, all pairs of nodes $u,v$ satisfy
    \begin{align}
        \left\lfloor \frac{d_E(u,v)}{r} \right\rfloor \leq d_G(u,v) \leq \left\lfloor \frac{d_E(u,v) + \kappa}{r}\right\rfloor,\label{eq:thm39}
    \end{align}
    where
    \begin{align}
        \kappa = c \left(\frac{d_E(u,v)}{r^{4/3}}+\frac{\log n}{r^{1/3}}\right).
    \end{align}
\end{theorem}

This theorem was proven for random geometric graphs on the square $[0, \sqrt{n}]^2$. The proof is based on a greedy algorithm in which the authors construct a \emph{greedy path} between the nodes $u$ and $v$. Starting from node $x_0= u$, they choose node $x_{i+1}$ as the neighbor of $x_{i}$ that has the smallest distance to node $v$. The authors then show that the length of this constructed path from $u$ to $v$ is at most the upper bound in \eqref{eq:thm39}. This upper bound is based on the probability that there is a node in the intersection of the ball with radius $r$ around $x_i$ and the ball with radius $a_t$ around $v$, where $a_t$ is chosen such that the intersection area is $\ln(2)$. Since none of the methods and concepts used in this proof are based on the boundary conditions of the square area, the theorem is not specific to geometric random graphs on a square but also holds when the underlying space is toroidal. \\

We rewrite Theorem \ref{thm:lowerbound_upperbound} to provide an upper bound and a lower bound for the Euclidean distance between nodes $u$ and $v$, given a constant $c$ and graph distance $t \coloneqq d_G(u,v)$. This yields
\begin{align}
    d_E(u,v) &\geq \frac{r\left(tr^{4/3} - c\log(n)\right)}{c + r^{4/3}} \eqqcolon \ell_n(t), \label{eq:lowerbound}\\
    d_E(u,v) &\leq r(t+1) \eqqcolon u_n(t), \label{eq:upperbound}
\end{align}
for $t \geq 1$. \\

We now introduce our main results about the query complexity of GRG reconstruction using the \textsc{Simple} algorithm \cite{mathieu2023simple}.

% \begin{theorem}\label{thm:main}
% Consider a connected geometric random graph $\mathcal{G}(n,r)$ on $[0,\sqrt{n}]^2$ with torus boundary conditions and connection radius $r \sim n^{k}$, $0< k < \frac{1}{2}$. The \textsc{Simple} algorithm gives with high probability a full reconstruction of $\mathcal{G}(n,r)$ in $O\!\left(n^{2k+1} \right)$ distance queries.
% \end{theorem}

\begin{theorem}\label{thm:main}
Consider a connected geometric random graph $\mathcal{G}(n,r)$ on $[0,\sqrt{n}]^2$ with torus boundary conditions and connection radius $r > 0$.
\begin{itemize}
    \item When $r \sim n^k$, $0 < k < \frac{1}{2}$, the \textsc{Simple} algorithm gives with high probability a full reconstruction of $\mathcal{G}(n,r)$ in $\tilde{O}\left(n^{2k+1} \right)$ distance queries if $\frac{3}{20} \leq k < \frac{1}{2}$ and $\tilde{O}\left(n^{3/2 - 4k/3}\right)$ distance queries if $0 < k < \frac{3}{20}$. 
    \item When $r = o(\sqrt{n}), r \geq C_1 \sqrt{\log(n)}$ for an absolute constant $C_1$, the \textsc{Simple} algorithm gives with high probability a full reconstruction of $\mathcal{G}(n,r)$ in $\tilde{O}\left(n^{3/2} \right)$ distance queries. 
\end{itemize}
\end{theorem}

With this theorem, we extend the capabilities of the \simple algorithm to a specific type of random graphs. We show that the \simple algorithm fully reconstructs dense geometric random graphs with growing degree in an almost optimal number of distance queries. Indeed, the GRG contains $\tilde{O}(n^{2k+1})$ edges in expectation. Compared to the guaranteed query complexity of $\tilde{O}(n^{5/3})$ of the \simple algorithm, our complexity bound for geometric random graphs shows an enormous improvement. For sparser GRGS with connection radius $r = O(\sqrt{\log(n)})$, we show that the \simple algorithm fully reconstructs the graphs in $\tilde{O}(n^{3/2})$. This result is again an improvement on the complexity bound of the \simple algorithm, and matches the complexity bound described in \cite{kannan2018graph} for bounded-degree graphs. Here, we give an outline of the proof of Theorem \ref{thm:main}. In Section \ref{sec:lemmas} we provide the necessary lemmas and give a proof. The proof is based on two parts:
\begin{enumerate} 
    \item For every node $u \in V$, we show that there exist at least four seeds that are at most Euclidean distance $d$ from node $u$ (Lemma \ref{lem:4seeds}), where the value of $d$ depends on $k$.
    \item We then show in Lemma \ref{lemma:simple_distinguishable} that after the initialization step in \simple, all node pairs that have Euclidean distance more than $r + 2 w_n(\tmax) \log(n)$ are distinguished by the seed set, where $\tmax$ is the maximal graph distance given Euclidean distance $d$. Consequently, in the second step of the algorithm only the pairs of nodes that are at Euclidean distance of at most $r+ 2\wmax \log(n)$ have to be queried.
\end{enumerate}

By combining the two observations, we can determine the number of queries that is needed to fully reconstruct the graph. This number is equal to the number of seed queries, i.e., the number of seeds times the number of nodes combined with the number of queries in the second step of the algorithm.\\

The results in Theorem \ref{thm:main} are only almost edge-optimal for dense GRGs as we rely in the proof on a number of seeds of order $n^\epsilon$, with $\epsilon < 1$ non-negative. 
% In a GRG with $r = o(\sqrt{n})$, the number of queries from all seeds to the rest of the nodes is then of a larger order compared to the number of edges and therefore not edge-optimal. 
However, our simulations show that in practice often only four seeds are enough to fully reconstruct the graph in an almost edge-optimal number of queries, both in sparse ($r = O(\sqrt{\log(n)})$) and in dense ($r \sim n^k$) GRGs (Figure \ref{fig:seed_size}). Although Theorem~\ref{thm:main} does not show full reconstruction in edge-optimal time for sparse graphs with $r = O(\sqrt{n})$, we can show that with only $n^{1+o(1)}\log(n)$ queries, the \simple algorithm detects at least $75\%$ of the non-edges with high probability:

\begin{theorem}\label{lem:connection_to_theorem}
     Consider a geometric random graph $\mathcal{G}(n,r)$ on $[0,\sqrt{n}]^2$ with torus boundary conditions and $r > C_1\sqrt{\log(n)}$ for a constant $C_1>0$. With a seed set $S$ consisting of $|S| = \log(n)n^{\epsilon}, 0 < \epsilon < 1$ randomly chosen nodes, the first step of the \simple algorithm can detect at least $75\%$ of the non-edges in expectation. 
\end{theorem}

%Instead of a random seed set, four seeds at specific seed locations suffice as to find at least $75\%$ of the non-edges. However, without knowledge on the node locations it is difficult to find these four seeds.\\ 

Here, we sketch an outline of the proof given in Section \ref{sec:non-edges}. We first define four \textit{optimal} locations for four seed nodes on the torus that maximize the distance between any two seeds in Lemma \ref{def:optimalseeds}. 
% Given these optimal locations, we calculate the number of distinguishable node pairs based on the number of nodes closer than $\ell_n(t)$ and farther than $u_n(t+1)$ to this seed for all possible $t$. 
% Since the seed nodes are as far from each other as possible, the overlap between these circles around the four seeds is minimized which enables us to count these distinguishable node pairs. 
In Lemma \ref{thm:non-edges} we show that we can detect at least $75\%$ of the non-edges given the presence of four seeds at the optimal seed locations. \\

Then, in Lemma \ref{lem:dense-regime_seeds}, we show that there is a seed in $S$ that is at most Euclidean distance $y$ apart from each optimal location when $|S| = n^{\epsilon}\log(n)$. Finally, combining Lemmas \ref{thm:non-edges} and \ref{lem:dense-regime_seeds} we show that these near-optimal seed locations result in at least $75\%$ of the non-edges detected.

\section{Proof of Theorem \ref{thm:main}}\label{sec:lemmas}
In this section, we prove Theorem \ref{thm:main} using the following two lemmas. As a seed set $S$, we choose a random sample of all nodes in $V$. The set $S$ can be regarded as a thinning of the Poisson point process underlying the GRG with resulting intensity $\lambda_s$ for some $\lambda_s<1$.\\

The following lemma shows that for all nodes in the GRG, we can find four seed nodes in $S$ that are at most at Euclidean distance $d$ from that respective node.

\begin{lemma}\label{lem:4seeds}
    Suppose that $\Phi$ is a Poisson point process on the torus $[0, \sqrt{n}]^2$ with intensity $1$ and that $S$ is a seed set consisting of a thinning of $\Phi$ with resulting intensity $\ls = \log(n)n^{\epsilon -1}$ for some $\epsilon \in (0, 1)$. Let $d = n^a$ for $a \in (0,1/2)$. For all $v \in V$, when $2a + \epsilon \geq 1$, with high probability there exists a set $S_v \subseteq S$ with $|S_v| \geq 4$ such that such that $d_E(v,s) \leq d$ for all $s \in S_v$ when $n \rightarrow \infty$. 
\end{lemma}

\begin{proof}
Let $v \in V$ and define $\B_v$ as the event that there are at least four seeds at Euclidean distance at most $d$ from node $v$. The probability of this event is
\begin{align}
    \P(\B_v) = \P(N_s( \pi d^2) \geq 4) &= 1-e^{-\ls \pi d^2} \sum_{i=0}^3 \frac{(\ls \pi d^2)^i}{i!} \label{eq:poisson}.
\end{align}
By the union bound, the probability that this event occurs for every point $v \in V$ is bounded by
\begin{align}
    \P\left(\bigcap_{v \in V}\B_v\right) = 1 - \P\left(\bigcup_{v \in V}\B^c_v\right) &\geq 1 - \sum_{v \in V}\P(\B^c_v) \\ 
    &= 1 - n e^{-\ls \pi d^2}\sum_{i=0}^3 \frac{(\ls \pi d^2)^i}{i!}\label{eq:Poisson_combined}.
\end{align}
As defined in the lemma, $d = n^a$ and $\ls = \log(n)n^{\epsilon -1}$. Therefore,
\begin{align}
    \P\left(\bigcap_{v \in V}\B_v\right) \geq 1 -  n^{1- \pi n^{2a + \epsilon - 1}}\sum_{i=0}^3 \frac{\left(\pi \log(n)n^{2a + \epsilon - 1}\right)^i}{i!}.\label{eq:asymptotics}
\end{align}

 When $n \rightarrow \infty$, \eqref{eq:asymptotics} tends to 1 as long as $2a + \epsilon \geq 1$. 
\end{proof}

In addition to guaranteeing that there are at least four seeds at Euclidean distance at most $d$ from every node $v \in V$, we also ensure that every seed node is located in another sector of the circle of radius $d$ around node $v$. Therefore, we specify the sector in which each of the four seeds (Figure \ref{fig:angles}) around the nodes in $V$ has to lie in the following corollary.

\begin{corollary}\label{cor:angles}
Lemma \ref{lem:4seeds} also holds for every $v \in V$ with the following extra condition on the configuration of the seeds $\{s_1, s_2, s_3, s_4\} \subseteq S_x$:
    \begin{align}
        \varphi(v,s_1) &\in \left[\frac{\pi}{2}  , \frac{3\pi}{4} + \alpha \right], \nonumber \\
        \varphi(v,s_2) &\in \left[\pi  , \frac{5\pi}{4} \right],\nonumber \\
        \varphi(v,s_3) &\in \left[\frac{3\pi}{2}, \frac{7\pi}{4} + \alpha \right],\nonumber \\
        \varphi(v,s_4) &\in \left[0 ,  \frac{\pi}{4} \right],\label{eq:angles}
    \end{align}
where $\varphi(v, s)$ is the angle between the horizontal axis and the line between node $v$ and seed $s$.
\end{corollary}

\begin{figure}[ht!]
    \centering
    \includegraphics[trim = 1.5cm 1.5cm 1.5cm 1.5cm,width=0.45\textwidth]{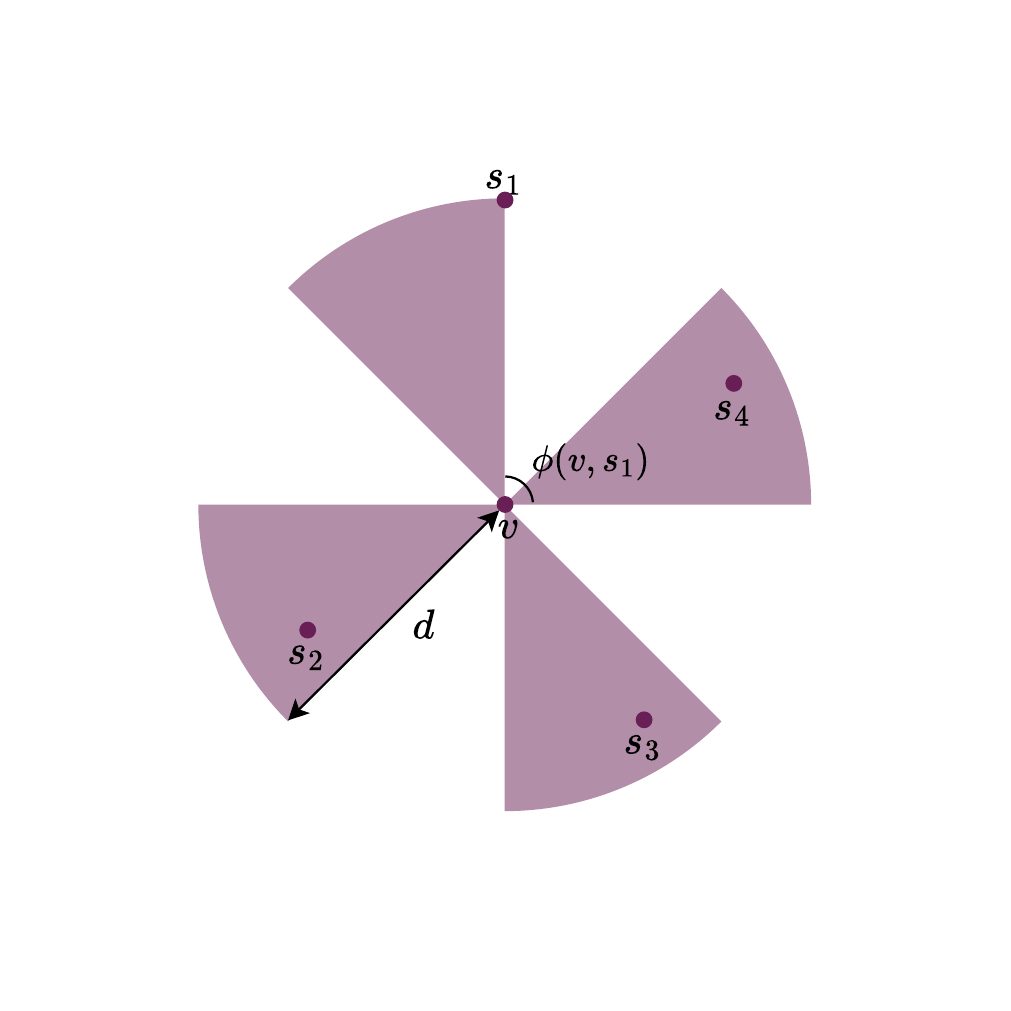}
    \caption{Regions in which the four seeds are located.}
    \label{fig:angles}
\end{figure}
\begin{proof}
As in Lemma \ref{lem:4seeds}, the seed nodes follow a PPP with intensity $\lambda_s$. However, instead of computing the probability that there are at least four seeds in a circle with radius $d$ around node $v$, we restrict these four seeds to be in a specific sector of the circle, where every sector has area $\pi/8 d^2$. We define $\tilde{\B}_v$ as the event that these four (disjoint) sectors of size $\pi/8 d^2$ all contain at least one seed:
\begin{align}
    \P\left((\tilde{\mathcal{B}}_v)^c\right) &= \P(\text{no four disjoint areas of size $\pi/8 d^2$ contain at least one seed}) \nonumber \\ 
    &= \P(\text{0, 1, 2 or 3 disjoint areas contain at least one seed})\nonumber \\ 
    &= \left(\P\left(N_s\left(\frac{\pi}{8}d^2\right)= 0\right)\right)^4 + {4\choose 3}\P(N_s\left(\frac{\pi}{8}d^2\right) \geq 1)\left(\P\left(N_s\left(\frac{\pi}{8}d^2\right) = 0\right)\right)^3  \nonumber \\ 
    &\hspace{1cm} + {4 \choose 2} \left(\P\left(N_s\left(\frac{\pi}{8}d^2\right) \geq 1\right)\right)^2\left(\P\left(N_s\left(\frac{\pi}{8}d^2\right) = 0\right)\right)^2  \nonumber \\
    &\hspace{1cm}+ {4 \choose 1} \left(\P\left(N_s\left(\frac{\pi}{8}d^2\right) \geq 1\right)\right)^3\P\left(N_s\left(\frac{\pi}{8}d^2\right) = 0\right) \nonumber \\
    &= e^{-\frac{\pi}{8} \ls d^2}\left(5 - 6 e^{-\frac{\pi}{8} \ls d^2} + 4e^{-\frac{\pi}{4} \ls d^2} - 2e^{-\frac{3\pi}{8}\ls d^2}\right).
\end{align}
Taking the union bound as in \eqref{eq:Poisson_combined} gives
\begin{align}
    \P\left(\bigcap_{v \in V}\tilde{\mathcal{B}}_v\right) &= 1 - \P\left(\bigcup_{v \in V}{\left(\tilde{\mathcal{B}}_v\right)^c}\right) \geq 1 -\sum_{x\in \Phi} \P\left(\left(\tilde{\mathcal{B}}_v\right)^c\right) \nonumber \\ 
    &= 1- ne^{- \frac{\pi}{8} \ls d^2}\left(5 - 6 e^{-\frac{\pi}{8} \ls d^2} + 4e^{-\frac{\pi}{4} \ls d^2} - 2e^{-\frac{3\pi}{8}\ls d^2}\right).
\end{align}
This probability tends to one as well when $2a + \epsilon \geq 1$. Thus, Lemma \ref{lem:4seeds} also holds for the given restriction on seed locations.
\end{proof}

Now, we show that every pair of node $u, v \in V$ has a distinguisher, which connects our previous observations to Line \ref{line:condition} of the \textsc{Simple} algorithm and gives an indication of the number of node pairs that are distinguishable given the seed set $S$.

\begin{lemma}\label{lemma:simple_distinguishable}
     For every pair of nodes $u, v\in V$ there is a seed $s \in S$ that can distinguish $u,v$ when $d_E(u,v) \geq 2w_n(\tmax)\log(n) + r$, where
     \begin{align}
        w_n(\tmax) &= \frac{c d}{r^{4/3}} + r + \frac{c \log(n)}{r^{1/3}},\label{eq:lemmaWmax}
     \end{align}
     for a constant $c>0$. We define $\tmax$ as the maximal graph distance between two nodes given Euclidean distance $d$
     \begin{align}
        \tmax = \frac{d(c+r^{4/3})}{r^{7/3}} + \frac{c\log(n)}{r^{4/3}}.\label{eq:tmax}
     \end{align}
\end{lemma}

% In Appendix \ref{app:lemma3}, we give a reasoning behind the choosing value $\wmax$, which is based on the area in which a node can still lie given the four surrounding seed nodes. 

\begin{proof}[Proof]
We choose a seed $s\in S$ that is Euclidean distance $\tilde{d} \leq d$ away from node $u$ and as far as possible from node $v$. We define $t_1 = d_G(u,s)$ and $t_2 = d_G(v, s)$, where $t_1 \leq t_2$. The seed $s$ distinguishes $u$ and $v$ if $|t_1 - t_2| > 1$, which means that $t_1 + 2 \leq t_2$. By using Theorem \ref{thm:lowerbound_upperbound} to bound $t_1$ and $t_2$ in terms of the Euclidean distances we obtain
\begin{align}
    t_1 &\leq \left\lfloor\frac{\tilde{d}}{r} + \frac{c \left(\frac{\tilde{d}}{r^{4/3}}+\frac{\log n}{r^{1/3}}\right)}{r}\right\rfloor = \frac{\tilde{d}}{r}\left(1 + \frac{c}{r^{4/3}}\right) + \frac{c\log(n)}{r^{4/3}}  \label{eq:upperboundt1},\\
    t_2 &\geq \left\lfloor\frac{d_E(s,v)}{r}\right\rfloor \geq \frac{d_E(s,v)}{r} - 1,\label{eq:lb2}
\end{align}
given a constant $c > 0$.

Now, we show that $t_1 + 2 \leq t_2$, which can be achieved by showing that $r(t_1+3) \leq d_E(s,v)$ in light of~\eqref{eq:lb2}. We first rewrite the left-hand side of this inequality to
\begin{align}
    r(t_1+3) &\leq \tilde{d}\left(1 + \frac{c}{r^{4/3}}\right) + \frac{c \log(n)}{r^{1/3}} + 3 r.\label{eq:rt1}
\end{align}

The Euclidean distance between node $v$ and seed $s$ is
\begin{align}
    d_E(s, v) = \sqrt{\tilde{d}^2 + (2\wmax \log(n) + r)^2 - 2\tilde{d}^2 (2\wmax \log(n) + r) \cos(\alpha)},
\end{align}
where $\alpha$ denotes the angle between the line $u,v$ and $u, s$ (Figure \ref{fig:worst_seed}). Since the four seeds are located in the four regions defined in Corollary \ref{cor:angles}, we compute an upper bound of $d_E(s, v)$ when $\alpha = \frac{5\pi}{8}$,
 \begin{align}
    d_E(s,v) &\leq \sqrt{\tilde{d}^2 + (2\wmax\log(n) + r)^2 - 2\tilde{d}(2\wmax\log(n) + r)\cos\left(\frac{5\pi}{8}\right)}\nonumber \\ 
    &= \sqrt{\tilde{d}^2 + (2\wmax\log(n) + r)^2 +2\tilde{d}(2\wmax\log(n) + r)\sin\left(\frac{\pi}{8}\right)}.\label{eq:dE} 
\end{align}

\begin{figure}[tbp]
    \centering
    \includegraphics[trim=2cm 2cm 2cm 2cm,width=0.4\textwidth]{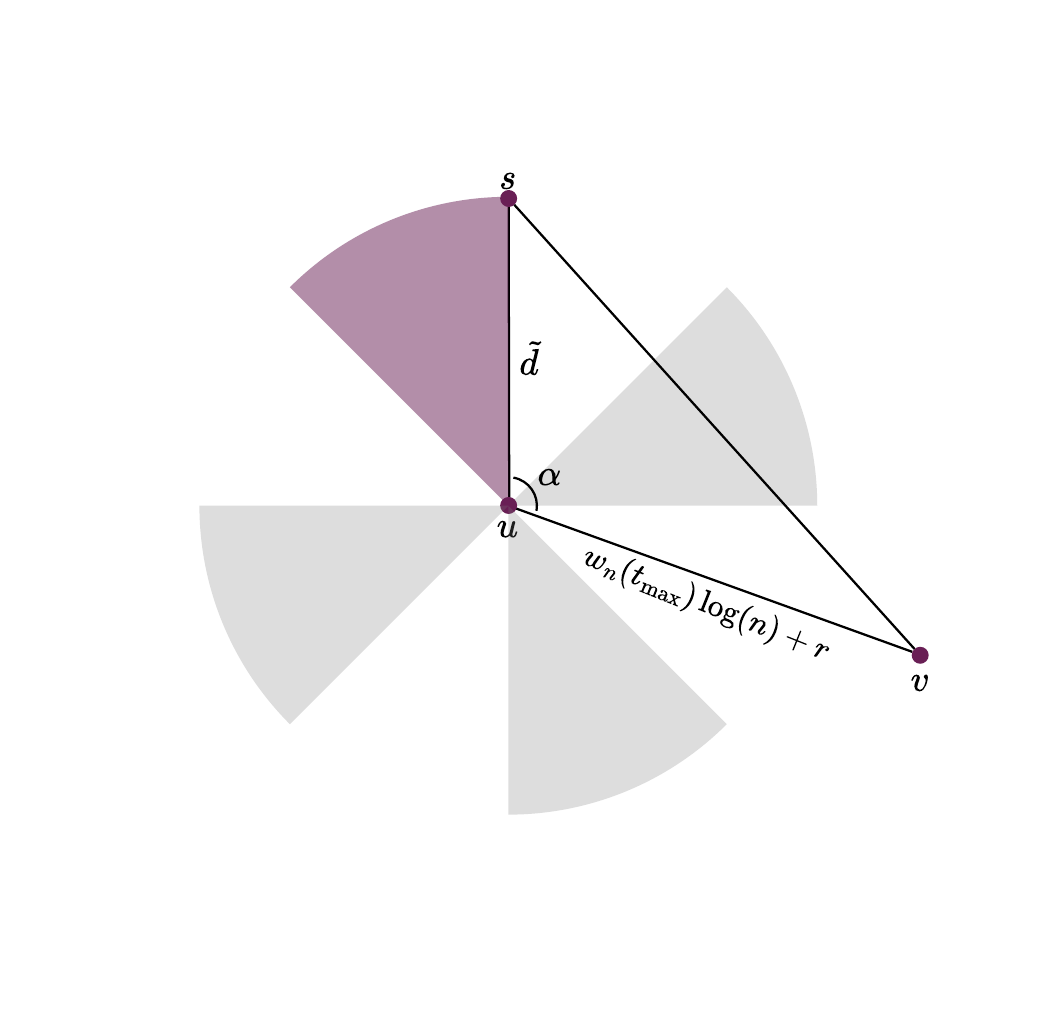}
    \caption{Worst case location of the furthest seed $s$ relative to $v$. }
    \label{fig:worst_seed}
\end{figure}
    % By the cosine rule, the Euclidean distance from node $v$ to seed $s$ is
   
% where we lower bound the cosine term by $-2^{-1/2}$ since $\alpha \in [0, \pi/2]$.

The difference between the square of \eqref{eq:dE} and \eqref{eq:rt1} is equal to
\begin{align}
    d_E(s,v)^2 - \left(r(t_1+3)\right)^2 &=4 r^2 \left(\log ^2(n)+\log (n)-2\right) +2\tilde{d} r \left(\sin \left(\frac{\pi }{8}\right) (2 \log (n)+1)-3\right)\nonumber \\ 
    &\hspace{0.5cm}+2 c r^{2/3} \left(4 \log ^2(n)+2 \log (n)-3\right) \log (n)\nonumber \\ 
    &\hspace{0.5cm}+ \frac{2 c \tilde{d}}{r^{1/3}} \left(2\log(n) \left(\frac{2d}{\tilde{d}}\log(n) + \sin\left(\frac{\pi}{8}\right) + \left(\frac{d}{\tilde{d}}-1\right)\right) -3\right)\nonumber \\
    &\hspace{0.5cm}+\frac{2 c \tilde{d}^2 }{r^{4/3}}\left(\frac{2d}{\tilde{d}} \sin \left(\frac{\pi }{8}\right) \log (n)-1\right)\nonumber\\
    &\hspace{0.5cm}+\frac{c^2 }{r^{8/3}}\big(4 d \log ^2(n)+h^2+2 h r \log (n) \left(\frac{4d}{\tilde{d}} \log ^2(n)-1\right)\nonumber \\
    &\hspace{2cm}+r^2 \log ^2(n)\left(4 \log ^2(n)-1\right) \big),\label{eq:difference}
    % &\geq 4 r^2 \left(\log ^2(n)+\log (n)-2\right) +d r \left(\sqrt{2} \log (n)+\frac{1}{\sqrt{2}}-6\right)\nonumber \\
    % &+ 2c  r^{2/3} \log (n) \left(4\log ^2(n)+ \log (n)-3\right) \nonumber \\
    % &+\frac{c d}{r^{1/3}} \left(\left(\sqrt{2}+8\right)  \log ^2(n)+2  \log (n)-6 \right)  +\frac{c d^2}{r^{4/3}} \left(\sqrt{2} \log (n)-2\right) \nonumber\\
    % &+ c^2 \left(\frac{\log ^2(n)}{r^{2/3}}+\frac{2 d \log(n)}{r^{5/3}}+ \frac{d^2}{r^{8/3}}\right) \left(4  \log ^2(n)-1\right).\label{eq:difference}
\end{align}
The difference in \eqref{eq:difference} is always positive for $n$ sufficiently large since $d/\tilde{d} \geq 1$, implying that $d_E(s,v) \geq r(t_1 + 3)$, and therefore $t_1 + 2 \leq t_2$. Thus, we have shown that all nodes that are at least at Euclidean distance $w_n(\tmax)\log(n) + r$ apart can be distinguished by a seed in $S$.
\end{proof}

We now have all building blocks to prove Theorem \ref{thm:main}.

\begin{proof}[Proof of Theorem \ref{thm:main}]
Let $\mathcal{G}(n,r)$ be a geometric random graph on $[0, \sqrt{n}]^2$ with torus boundary conditions and $r \sim n^k$ , $k \in (0, 1/2)$. By Lemmas \ref{lem:4seeds} and \ref{lemma:simple_distinguishable}, all node pairs that have Euclidean distance larger than $2\wmax \log(n)+ r$ are distinguishable, based on a seed set of $|S| = \log(n)n^{\epsilon}$ randomly-chosen nodes, with $\epsilon \in (0, 1)$. Thus, the set of non-distinguishable node pairs that are queried in the second step of the \simple algorithm (Line \ref{line:second}) can be upper bounded by the set of all node pairs nodes that are at most Euclidean distance $2w_n(\tmax)\log(n) + r $ apart.\\ 

Let us determine the number of queries that still have to be performed. We define the event $\mathcal{C}_v:= \{N_c \geq (1+\delta)\E(N_c)\}$, where $N_c$ is equal to the number of nodes in a circle around node $v$ with radius $2w_n(\tmax)\log(n)+r$. Using the Chernoff bound, we have:
\begin{align}
    \P(\mathcal{C}_v)&=\P(N_c \geq (1+\delta)\E(N_c)) \leq  e^{-\delta^2 \E(N_c)/2},\\ 
    \E(N_c) &= \pi \left(2\wmax\log(n) + r\right)^2,\label{eq:expectation}
\end{align}
with $\wmax$ as given in \eqref{eq:lemmaWmax}.
By the union bound, the probability that this holds for all $v \in V$ is bounded by:
\begin{align}
    \P\left(\bigcup_{v\in V} \mathcal{C}_v \right) &\leq n \P(N_c \geq (1+\delta)\E(N_c)) \leq n e^{-\delta^2 \E(N_c)/2} \nonumber \\ 
    &\leq n^{1-\delta^2\pi \wmax}, \label{eq:chernoff} 
\end{align}
where the last line follows as an upper bound since $\E(N_c) > 2 \pi \wmax \log(n) > 1$. Therefore, \eqref{eq:chernoff} tends to $0$ when $n \rightarrow \infty$ for any constant $\delta > 0$. Thus, with high probability, there are at most $(1+\delta) \E(N_c)$ nodes at Euclidean distance $2\wmax\log(n) + r$.\\

Finally, we compute the number of queries to fully reconstruct the graph using the \textsc{Simple} algorithm. For every seed, we have to query all nodes, resulting in a total of $n \times \log(n) n^{\epsilon}$ queries (Lemma \ref{lem:4seeds}). Then, in the second step of the algorithm, we only query the indistinguishable node pairs, of which there are at most $(1+\delta)\E(N_c)$ with high probability. Thus, the second step uses at most $n \times (1+\delta)\E(N_c)$ queries with high probability, which is of order
\begin{align}
    O\left(n \times (1+\delta) \pi \left(2\wmax \log(n)+r)^2\right)\right) = 
    \tilde{O}\left(\max\left\{\frac{n^{2a+1}}{r^{8/3}},\frac{n^{a+1}}{r^{1/3}}, nr^{2} \right\}\right),\label{eq:querystep2}
\end{align}
filling in \eqref{eq:lemmaWmax} and $d = n^a$.\\

The total query complexity of the \simple algorithm
%, that results in full reconstruction of the graph with high probability, 
is equal to the number of initial queries from the seed set combined with the number of queries needed to distinguish all node pairs, which gives
\begin{align}
    \tilde{O}\left(n^{1+\epsilon}\right) {+} \tilde{O}\left(\max\left\{\frac{n^{2a+1}}{r^{8/3}}, \frac{n^{a+1}}{r^{1/3}}, nr^{2} \right\}\right) {=} \tilde{O}\left(\max\left\{n^{1+\epsilon}, \frac{n^{2a+1}}{r^{8/3}}, \frac{n^{a+1}}{r^{1/3}}, nr^{2} \right\}\right).\label{eq:totalqueries}
\end{align}

When $r = o(\sqrt{n})$, \eqref{eq:totalqueries} simplifies to:
\begin{align}
    \tilde{O}\left(\max\left\{n^{1+\epsilon}, n^{2a+1}\right\}\right) = \tilde{O}\left(n^{3/2}\right),
\end{align}
for $a = 1/4$ and $\epsilon = 1/2$, satisfying $2a + \epsilon \geq 1$. In the dense regime with $r \sim n^k, 0 < k < \frac{1}{2}$, we have

\begin{align}
    \tilde{O}\left(\max\left\{n^{1+\epsilon}, \frac{n^{2a+1}}{r^{8/3}}, \frac{n^{a+1}}{r^{1/3}}, nr^{2} \right\}\right) &= \tilde{O}\left(\max\left\{n^{1+\epsilon}, n^{2a+1-8k/3}, n^{a+1-5k/3}, n^{2k+1} \right\}\right),\nonumber \\
    &=
    \begin{cases}
        \tilde{O}\left(n^{3/2 - 4k/3}\right), \hspace{0.8cm} &\text{if } 0 < k \leq \frac{3}{20},\\ 
        \tilde{O}\left(n^{2k+1}\right), & \text{if } \frac{3}{20} < k < \frac{1}{2},
    \end{cases}
\end{align}
with the following values for $a$ and $\epsilon$:
\begin{align}
    a = 
    \begin{cases} 
        \frac{1}{4} + \frac{2}{3}k, \hspace{0.8cm} &\text{if } 0 < k < \frac{3}{20},\\
        \frac{5}{12} - \frac{4}{9}k &\text{if } \frac{3}{20} < k < \frac{1}{2},
    \end{cases}
    \hspace{0.2cm}\epsilon = 
    \begin{cases} 
        \frac{1}{2} - \frac{4}{3}k, \hspace{0.8cm} &\text{if } 0 < k < \frac{3}{20},\\
        \frac{1}{12} + \frac{13}{9}k &\text{if } \frac{3}{20} < k < \frac{1}{2}.
    \end{cases}\label{eq:epsilon}
\end{align}
\end{proof}

\section{Proof of Theorem \ref{lem:connection_to_theorem}}\label{sec:non-edges}
The proof of Theorem \ref{lem:connection_to_theorem} is based on detecting the non-edges by finding all \textit{distinguishable} node pairs given four optimal seed locations, which we introduce in the following lemma.

\begin{figure}[tbh]
    \centering
    \includegraphics[width = 0.5\textwidth]{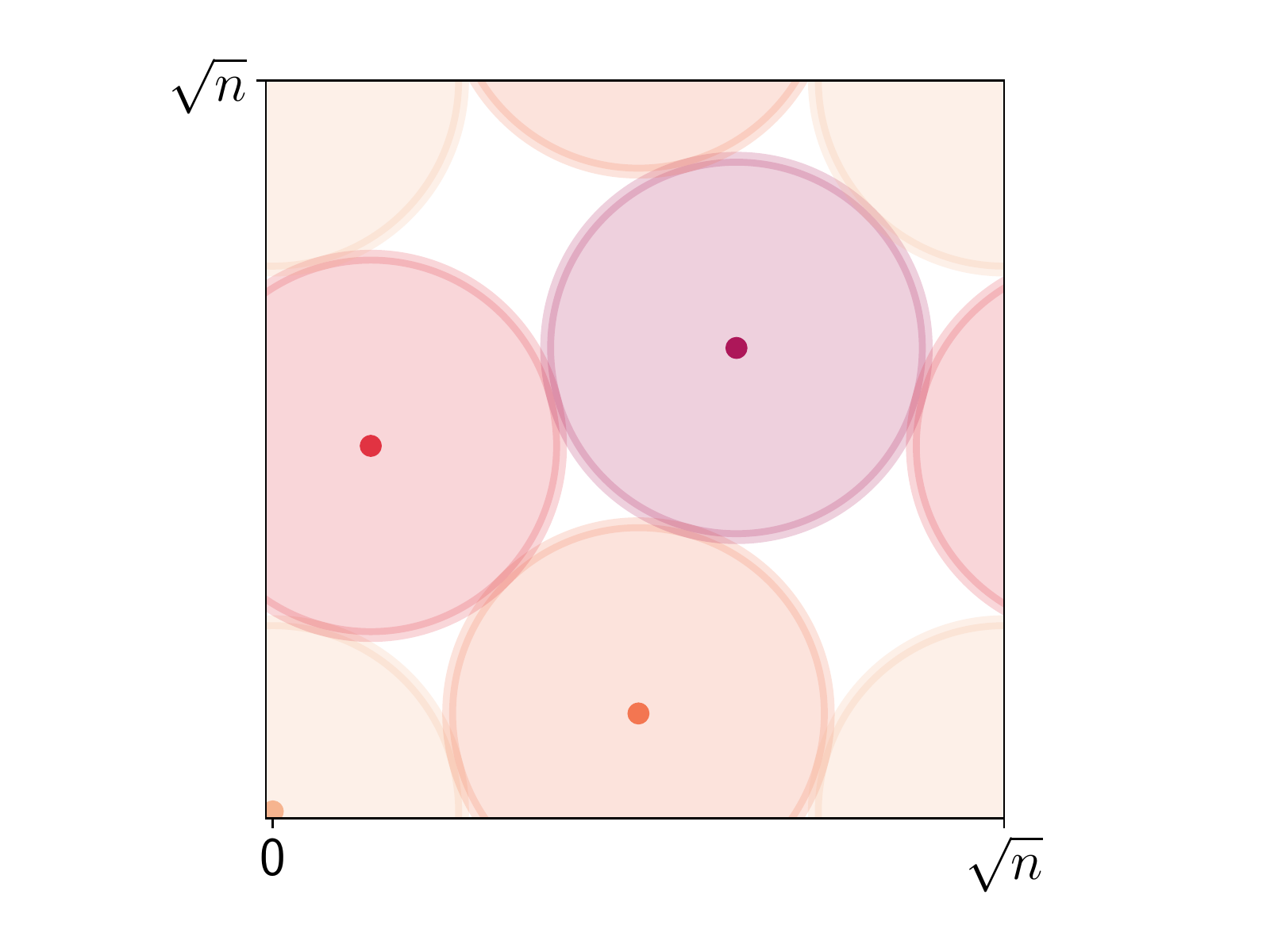}
    \caption{Optimal circle packing for 4 circles on a flat torus.}
    \label{fig:circle_packing}
\end{figure}

\begin{lemma}[Optimal seed locations \cite{dickinson2011optimal}]\label{def:optimalseeds}
The following four seed nodes $\Sopt = \{\tilde{s}_1, \tilde{s}_2, \tilde{s}_3, \tilde{s}_4\}$ with corresponding locations
\begin{align}
    \tilde{s}_1 &= (0, 0),\label{eq:s1}\\ 
    \tilde{s}_2 &= \left(\frac{1}{2}\sqrt{n}, \frac{1}{2}\sqrt{n(7-4\sqrt{3})}\right),\\
    \tilde{s}_3 &= \left( \frac{1}{2}\sqrt{n}(\sqrt{3}-2), \frac{1}{2}\sqrt{n}\right),\\
    \tilde{s}_4 &= \left(\frac{1}{2}\sqrt{3n}(\sqrt{3}-1),\frac{1}{2}\sqrt{3n}(\sqrt{3}-1) \right)\label{eq:s4}
\end{align}
are called \textit{optimal seed locations}, since the minimum distance between any of the four seed nodes is maximized and equal to:
\begin{align}
    x_n = \frac{\sqrt{6}-\sqrt{2}}{2}\sqrt{n}.\label{eq:optimal_diameter}
\end{align}
\end{lemma}

Given a seed node $s \in S$ and a $t > 1$, we classify two sets of nodes: (1) nodes for which the graph distance to $s$ will always be \emph{smaller} than $t$ and (2) nodes for which the graph distance to $s$ will always be \emph{larger} than $t + 1$. Consequently, any node pair with one node from both sets will be distinguishable by seed $s$. We call these two sets the lower and upper bound set and define these using \eqref{eq:lowerbound} - \eqref{eq:upperbound}:
\begin{definition}[Lower and upper bound set]\label{def:lower_upper_boundset}
    The lower bound set $L_s(t)$ of seed $s$ given the graph distance $t$ is the set of nodes in which the graph distance to seed $s$ is smaller than $t$,
    \begin{align}
        L_s(t) &= \{v \,|\, d_E(v, s) < \ell_n(t) \}. \label{eq:lowerboundset}
    \end{align}
    Similarly, the upper bound set $U_s(t)$ of seed $s$ given graph distance $t$ is the set of nodes in which the graph distance to seed $s$ is larger than $t+1$,
    \begin{align}
        U_s(t) &= \{v \,|\, d_E(v, s) > u_n(t+1) \}.
    \end{align}
\end{definition}

\begin{lemma}\label{thm:non-edges}
    Consider a connected geometric random graph $\mathcal{G}(n,r)$ on $[0,\sqrt{n}]^2$ with torus boundary conditions and four seed nodes added at the locations of Definition \ref{def:optimalseeds}. Querying these seeds with all nodes results in recovery of at least $75\%$ of the non-edges in expectation, for all $r >0$.
\end{lemma}

\begin{proof}
Given four seeds $S^{\text{opt}} = \{\tilde{s}_1, \tilde{s}_2, \tilde{s}_3, \tilde{s}_4\}$ with optimal seed locations given by Lemma \ref{def:optimalseeds}. By Definition \ref{def:distinguishing}, a pair of nodes $u, v \in V$ is \textit{distinguishable} by $s \in \Sopt$ when $|d_G(u, s) - d_G(v, s)| >1$. Combining this with the lower and upper bound set of definition \ref{def:lower_upper_boundset}, we define the set $\bar{E}_s(t) = \{(u, v) | (u,v) \in L_s(t) \times U_s(t) \}$ as the set of \textit{distinguishable} node pairs for graph distance $t$ and seed $s$. Note that this set is \textit{directional}, i.e. we define $(u,v)$ as a different node pair than $(v, u)$, for ease of notation and analysis. As a consequence, there are in total $n(n-1)$ possible node pairs.  \\

We are interested in the number of all distinguishable node pairs, the \textit{non-edges} of the graph, which equals the size of the union of the sets $\bar{E}_s(t)$ for all $t$ and $s \in S$: 
\begin{align}
    \left| \bigcup_t \bigcup_{s \in \Sopt} \bar{E}_s(t)\right| &=  \left| \bigcup_t  \bigcup_{s \in \Sopt} L_s(t) \times U_s(t)  \right| \nonumber \\
    &=  \sum_{t} \left| \bigcup_{s \in \Sopt} \left(L_s(t) \backslash L_s(t-1) \right) \times U_s(t)\right|
    \nonumber \\
    &=  \sum_{t} \left| \bigcup_{s \in \Sopt} \bar{L}_s(t) \times U_s(t)\right|.\label{eq:union_nonedges}
\end{align}
In the second step we have split the lower bound sets into rings to ensure that these sets do not overlap for different $t$. To simplify notation, we define these rings as $\bar{L}_s(t) \coloneqq L_s(t)\backslash L_s(t-1)$.\\

% \lotte{We split up this sum in two: first, we sum until $t = \tilde{t}_1$ such tat $\cap_{s \in \Sopt}\bar{E}_s(t) = \emptyset$. For the second part, we allow overlap between two lower bound sets but ensure that there is no overlap between the lower bound sets of three seeds.}

We choose a maximum graph distance $\tt$ such that
\begin{align}
    2 \ell_n(\tt_1) = x_n, \label{eq:find_t}
\end{align}
with $x_n$ as defined in \eqref{eq:optimal_diameter}. Solving \eqref{eq:find_t} gives:
\begin{align}
    \tt_1 = \frac{4c r \log (n)+ \sqrt{n}\left(\sqrt{6} - \sqrt{2}\right) \left(c+ r^{4/3}\right)}{4 r^{7/3}}\label{eq:tildet}
\end{align}

As the sets $\bar{E}_s(t)$ are disjoint for all $s$ as long as $t \leq \tt$, 
\begin{align}
   \sum_{t=1}^{\lfloor\tt_1\rfloor} \left| \bigcup_{s \in \Sopt} \bar{L}_s(t) \times U_s(t)\right| &= 4\sum_{t=1}^{\ftt} \left|\bar{L}_{s_1}(t) \times U_{s_1}(t)\right|   \nonumber \\ 
   &= 4\sum_{t=1}^{\ftt}  \left|\bar{L}_{s_1}(t) \right| \cdot \left|U_{s_1}(t) \right|.\label{eq:eq1}
\end{align}

Since nodes are distributed as a Poisson point process, the number of nodes in disjoint regions are independent. Moreover, we define $\ell_n(0) = 0$. Now, we can bound the expected number of non-distinguishable node pairs by
\begin{align}
    \E(|\bar{E}|)  &\geq 4 \sum_{t=1}^{\ftt}   \left(\pi^2 \left(\ell_n(t)^2 - \ell_n(t-1)^2\right) u_n(t+1)^2\right) \nonumber \\
    &\geq \frac{2 \pi  (\tt-1) r^{10/3} }{3 \left(c+r^{4/3}\right)^2}\big(2 \pi  c \left(2 \tt^2+11 \tt+24\right) r^2 \log (n)-12 c n \log (n)\nonumber \\
    &\hspace{4cm}-\pi  \left(3 \tt^3+11 \tt^2+5 \tt-24\right) r^{10/3}+6 n (\tt-1) r^{4/3}\big),
\end{align}
where the inequality follows since $\ftt \geq \tt - 1$. Finally, filling in $\tt$ as \eqref{eq:tildet} yields

\begin{align}
    \E\left(\left| \bar{E}\right|\right) &\geq \frac{\pi \left(16-8\sqrt{3} - 7\pi + 4\sqrt{3}\pi\right)r^{8/3}}{8\left(c+r^{4/3}\right)^2}n^2 + o(n^2), \label{eq:eq2}
\end{align}
since $r < \sqrt{n}$. When $n \rightarrow \infty$ (and so does $r$), the percentage of node pairs that is distinguishable is
\begin{align}
    \frac{\E(|\bar{E}|)}{n(n-1)} \xrightarrow{n\rightarrow\infty} \frac{\pi (16-8\sqrt{3} - 7\pi + 4\sqrt{3}\pi)}{8}  \approx 0.753.\label{eq:limit}
\end{align}

Thus, the number of non-edges that we detect after querying all nodes in the graph with four seeds equals at least $75\%$ of the total number of non-edges.    
    
\end{proof}
 
In Lemma \ref{thm:non-edges} we assume that the seed nodes are at specific locations. However, as we do not know the locations of the nodes, such seeds may not exist, and even if they do, we are not able to select them. In the following lemma we show, similarly to Lemma \ref{lem:4seeds}, that with in total $n^\epsilon$ seeds, there are four seeds that are at most distance $y$ away from the four specified optimal locations. 

\begin{lemma}\label{lem:dense-regime_seeds}
    Suppose that $\Phi_s$ is a Poisson point process on $[0, \sqrt{n}]^2$ with intensity $\lambda_s = \log(n)n^{\epsilon - 1}$. We take $y = n^b$ for $b \in (0,1/2)$. When $2b + \epsilon \geq 1$, with high probability there exist $s_i \in S$ with $i = 1,2, 3, 4$, such that $d_E(s_i,\tilde{s_i}) \leq y$ and $n \rightarrow \infty$ for $\tilde{s_i} \in \tilde{S}$.
\end{lemma}

\begin{proof}
We define $\B_{\tilde{s}}$ as the event that there is at least one point $s \in S$ at Euclidean distance at most $y$ from the optimal seed $\tilde{s} \in \tilde{S}$. The probability of this event is
\begin{align}
    \P(\B_s) = \P(N_s( \pi y^2) \geq 1) &= 1-e^{-\ls \pi y^2} \label{eq:poissonS}.
\end{align}
By the union bound, the probability that this event occurs for all four points $\tilde{s} \in \tilde{S}$ is bounded by
\begin{align}
    \P\left(\bigcap_{\tilde{s} \in \tilde{S}}\B_{\tilde{s}}\right) = 1 - \P\left(\bigcup_{\tilde{s} \in \tilde{S}}\bar{\mathcal{B}}_{\tilde{s}}\right) &\geq 1 - \sum_{\tilde{s} \in \tilde{S}}\P(\bar{\mathcal{B}}_{\tilde{s}}) = 1 - 4 e^{-\ls \pi y^2}\label{eq:Poisson_combinedS}.
\end{align}
As $y = n^b$ and $\ls = \log(n)n^{\epsilon -1}$,
\begin{align}
    \P\left(\bigcap_{\tilde{s} \in \tilde{S}}B_v\right) \geq 1 -  4n^{-\pi n^{2b+\epsilon -1} }.\label{eq:asymptoticsS}
\end{align}
When $n \rightarrow \infty$, \eqref{eq:asymptoticsS} tends to 1 as long as $2b + \epsilon \geq 1$, so these seeds exist with high probability. The circles with radius $y$ around the four optimal seed locations will not overlap, so $|S| \geq 4$.
\end{proof}

Now, we combine the above lemma's to prove Theorem \ref{lem:connection_to_theorem}.

\begin{proof}[Proof of Lemma \ref{lem:connection_to_theorem}]
    In Theorem \ref{thm:non-edges} we have shown that when there are four seeds at the optimal locations, we can reconstruct at least $75\%$ of the non-edges. We now have a seed set $S$ that is random sample of size $\log(n)n^\epsilon, \epsilon \in (0, 1)$ of the node set and we define $S^{(4)} = \{s_1, s_2, s_3, s_4\} \subseteq S$ as the four seed nodes that are closest to the optimal seed locations from Lemma \ref{def:optimalseeds}. By Lemma \ref{lem:dense-regime_seeds}, we know that with high probability the four seeds $S^{(4)}$ are at most Euclidean distance $y = n^b$ away from the optimal seed locations, where $2b + \epsilon \geq 1$. Similar to Theorem \ref{thm:non-edges}, we compute the number of distinguishable node pairs $|\bar{E}|$ as
    \begin{align}
        |\bar{E}| = \left| \bigcup_t \bigcup_{s \in S} \bar{E}_s(t)\right| &\geq \left| \bigcup_t \bigcup_{s \in S^{(4)}} \bar{E}_s(t)\right| \nonumber \\
        &= \sum_{t} \left| \bigcup_{s \in \Sopt} \bar{L}_s(t) \times U_s(t)\right|.
    \end{align}
    Again, $\cap_{s\in S^{(4)}}\bar{E}_s(t)$ has to be empty, which means in this setting that we choose the maximal graph distance $\tilde{t}$ such that
    \begin{align}
        u_n(\tt + 1) + \ell_n(\tt) = x_n - 2y,
    \end{align}
    since two seeds in $S^{(4)}$ are always Euclidean distance at least $x_n - 2y$. Then 
    \begin{align}
        \tt = \frac{2 c r \log (n)+ \left(\sqrt{n}(\sqrt{6} - \sqrt{2}) -4(r+2y)\right)\left(c+ r^{4/3}\right)}{2 \left(c r+2 r^{7/3}\right)}.
    \end{align}
    Then, we can follow the same steps as in \eqref{eq:eq1} -- \eqref{eq:eq2}, resulting in the same limit as in \eqref{eq:limit}, provided that $y = n^b < \frac{1}{2}$ and $r < \sqrt{n}$. Thus with $n \times |S| = n^{\epsilon+1}\log(n)$ queries we can already reconstruct at least $75\%$ of the non-edges in expectation. 
\end{proof}

\section{Simulations}
We now compare our theoretical results to simulations and explore the performance of the algorithm in practice. For all simulations, we show the average results over 100 iterations. The dashed lines in the plots show the theoretical bound on the complexity (up to a constant factor).\\

\begin{figure}[b!]
    \centering
    \includegraphics[width = 0.45\textwidth]{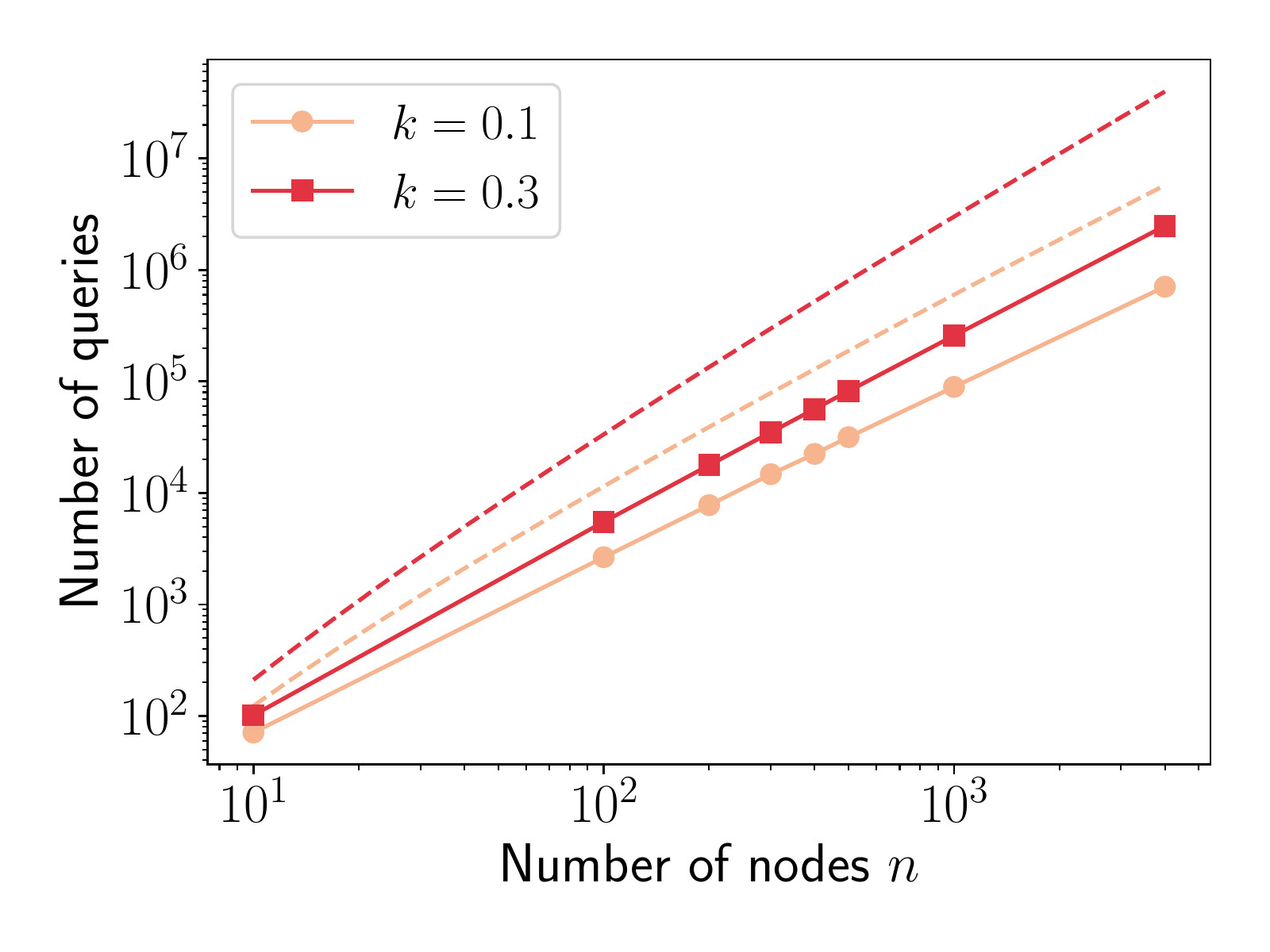}
    \caption{Simulated query complexity of the \simple algorithm for a GRG with $r \sim n^{k}$, $k = 0.1$ and $k = 0.3$. The upper dashed line represents the theoretical query complexity (up to a constant) from Theorem \ref{thm:main}.}
    \label{fig:differentk}
\end{figure}
First, we compare the query complexity to the theoretical bound for a dense GRG with $r = n^k$, where $k = 0.1$ and $k = 0.3$. Figure \ref{fig:differentk} shows that the theoretical query complexity is indeed an upper bound on the simulated query complexity. Moreover, in practice, the growth rate of the query complexity of the \simple algorithm is closer to the number of edges than to $\tilde{O}(n^{2k+1})$ or $\tilde{O}(n^{3/2-4k/3})$. Thus, the \simple algorithm indeed results in almost optimal reconstruction. \\

\begin{figure}[t!]
    \centering
    \includegraphics[width = 0.45\textwidth]{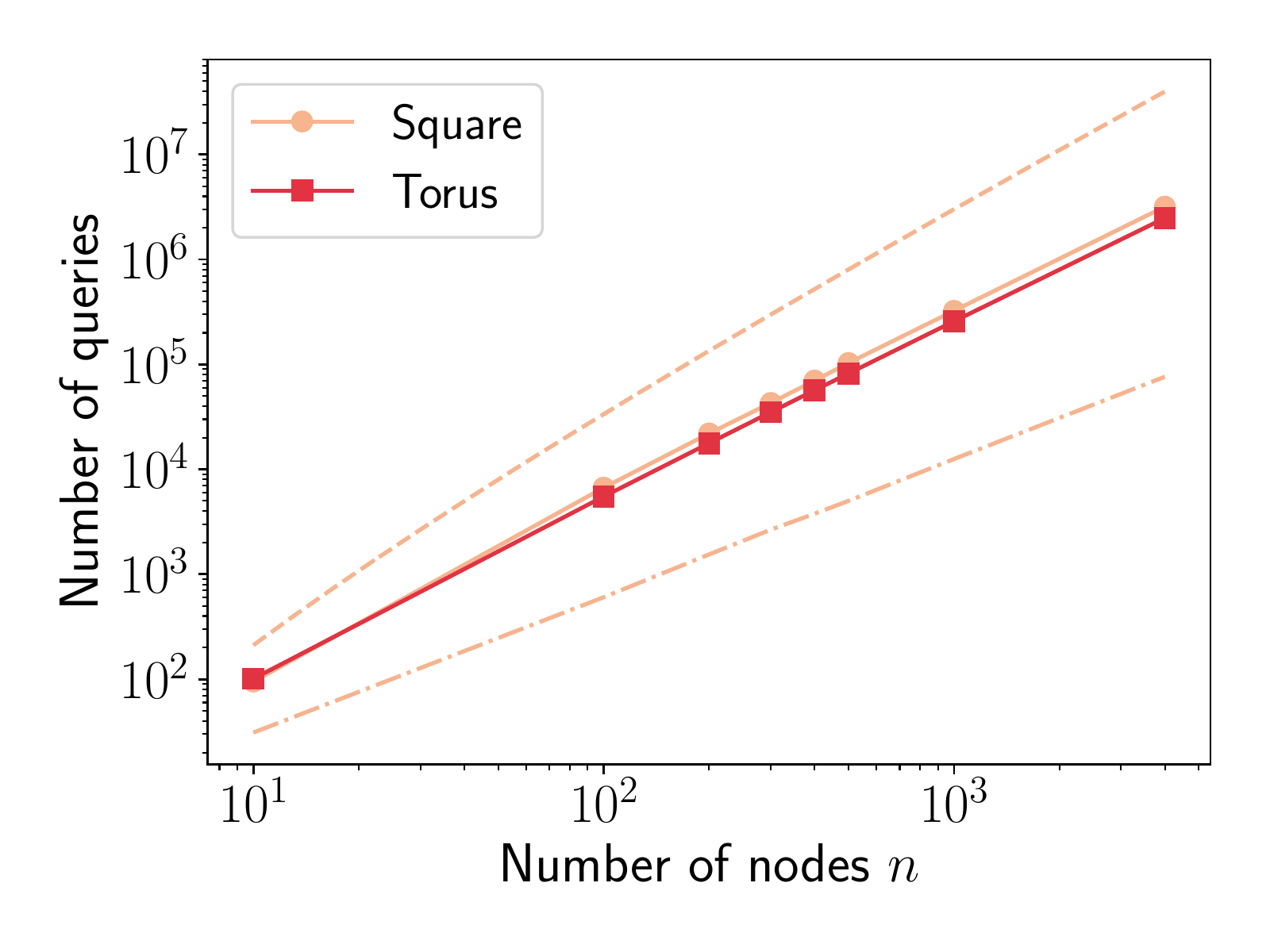}
    \caption{Simulated query complexity of the \simple algorithm for a GRG with $r \sim n^{0.3}$ on a torus and on a square. The upper dashed line represents the theoretical query complexity $n^{2k+1}$ (up to a constant) from Theorem \ref{thm:main} and the lower dashed line is the number of edges in the toroidal GRG.}
    \label{fig:torus}
\end{figure}

To simplify the proof, we have assumed torus boundary conditions for the GRG. Figure~\ref{fig:torus} shows the difference in query complexity for a GRG on a torus and a GRG on a square. The GRG on the torus performs slightly better in terms of query complexity compared to the GRG on a square. This is explained by the fact that a torus does not have a boundary and estimating graph distances between nodes close to the boundary is hard \cite{dani2023reconstruction} than far from the border. However, both on the square and on the torus, the query complexity follows our complexity bound. Therefore, it is reasonable to assume that the \simple algorithm has a similar query complexity for GRGs on a square. With some extra work, one could prove this result similarly to the proof of Theorem \ref{thm:main}. \\

\begin{figure}[b!]
    \centering
    \includegraphics[width = 0.45\textwidth]{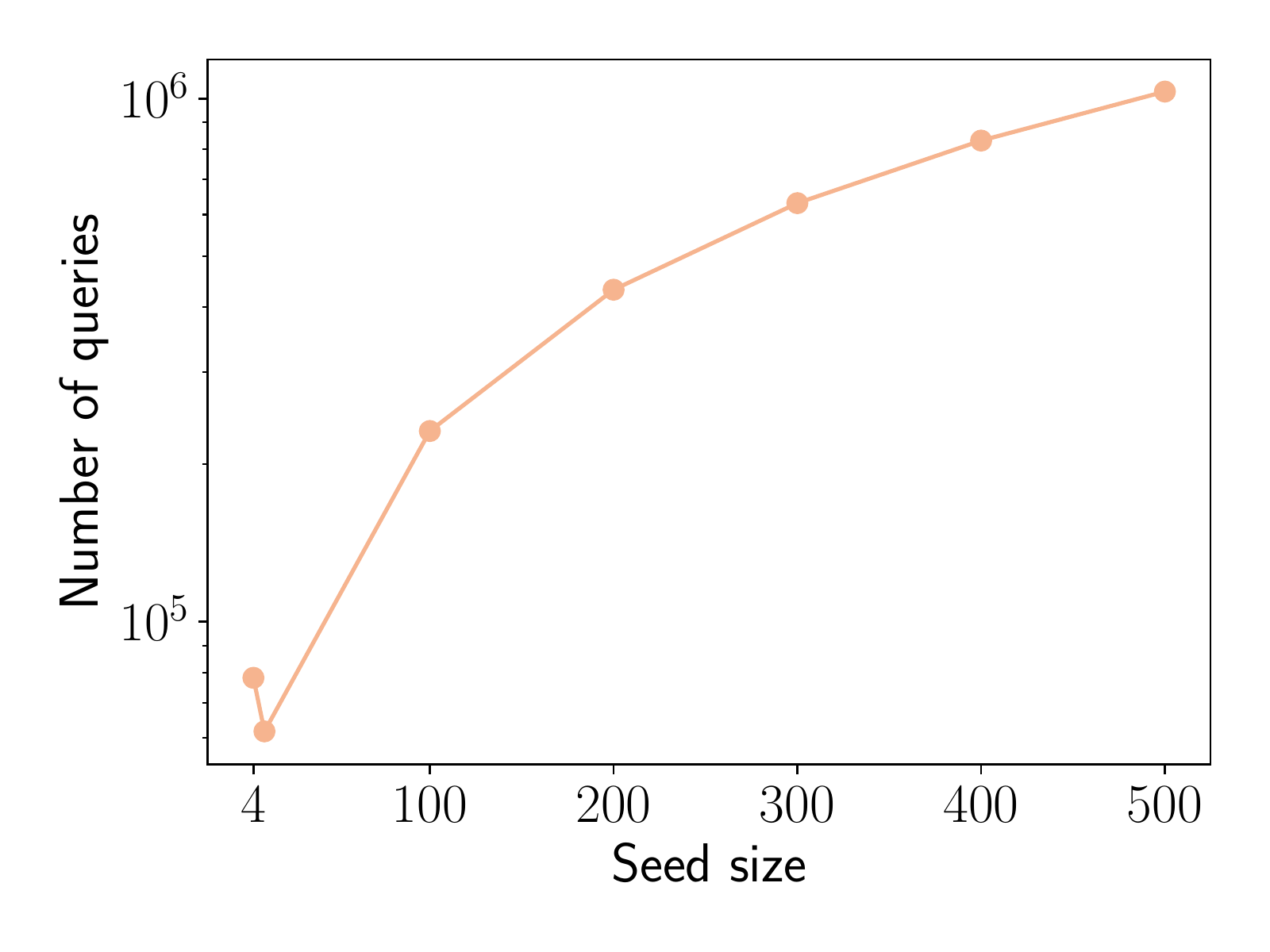}
    \caption{Simulated query complexity of the \simple algorithm for different seed sizes, based on a GRG with $k = 0.3$ and $n = 2000$. The seed size used in Theorem \ref{thm:main} corresponds to $|S| = 386$. }
    % The error bars show the minimum and maximum values over 100 iterations. }
    \label{fig:seed_size}
\end{figure}

Figure \ref{fig:seed_size} plots the query complexity of \simple for different seed set sizes. In our results, we use a seed set of size $\log(n)n^\epsilon$, with $\epsilon$ as defined in \eqref{eq:epsilon}. However, the smallest number of queries is obtained with a smaller seed size. This observation confirms our hypothesis that it is possible to reconstruct the graph in even less shortest path queries with only a few seeds.\\

In Figure~\ref{fig:RRG} we compare the \simple algorithm for a random regular graph and the GRG, where we choose the degree for the random regular graph such that it is equal to the average degree in the GRG. As the proven query complexity in \cite{mathieu2023simple} holds for sparse graphs while Theorem \ref{thm:main} only holds for dense geometric random graphs, we compare the RRG and the GRG both in the sparse regime (Figure \ref{sfig:sparse_comparison}), and in the dense regime (Figure \ref{sfig:dense_comparison}).\\

Note that the two plots in Figure \ref{fig:RRG} show different comparisons: in the sparse regime, for both the RRG and GRG we keep the degree constant at $\Delta = 50$, which means that $r$ in the GRG is variable, while in the dense regime, we fix the radius $r = n^k$ and change the degree of the RRG accordingly. Figure \ref{sfig:dense_comparison} shows that indeed, the GRG outperforms the random regular graph in query complexity in a dense setting. This trend is also observed in a sparse setting in Figure \ref{sfig:sparse_comparison}, indicating that even in a sparse setting with bounded degree one could prove a lower query complexity than $\tilde{O}(n^{5/3})$ in the case of geometric random graphs.\\

\begin{figure}[b!]
    \begin{subfigure}{ 0.45\textwidth}
        \includegraphics[width = \textwidth]{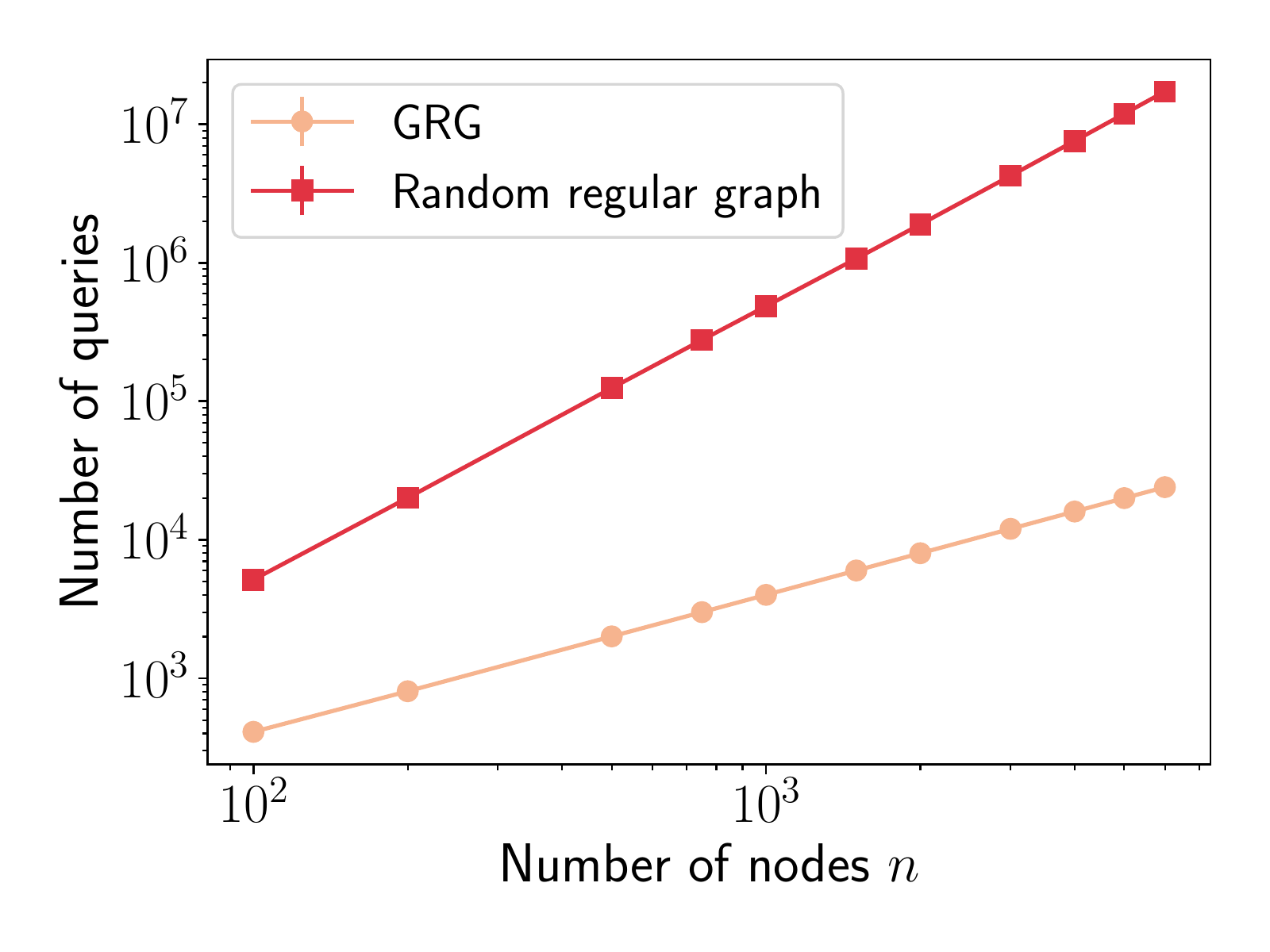}
        \caption{$r=\sqrt{\Delta/(n\pi)}$, where $\Delta = 50$ and $|S| = 4$.}
        \label{sfig:sparse_comparison}
    \end{subfigure}\hfill
       \begin{subfigure}{ 0.45\textwidth}
        \includegraphics[width = \textwidth]{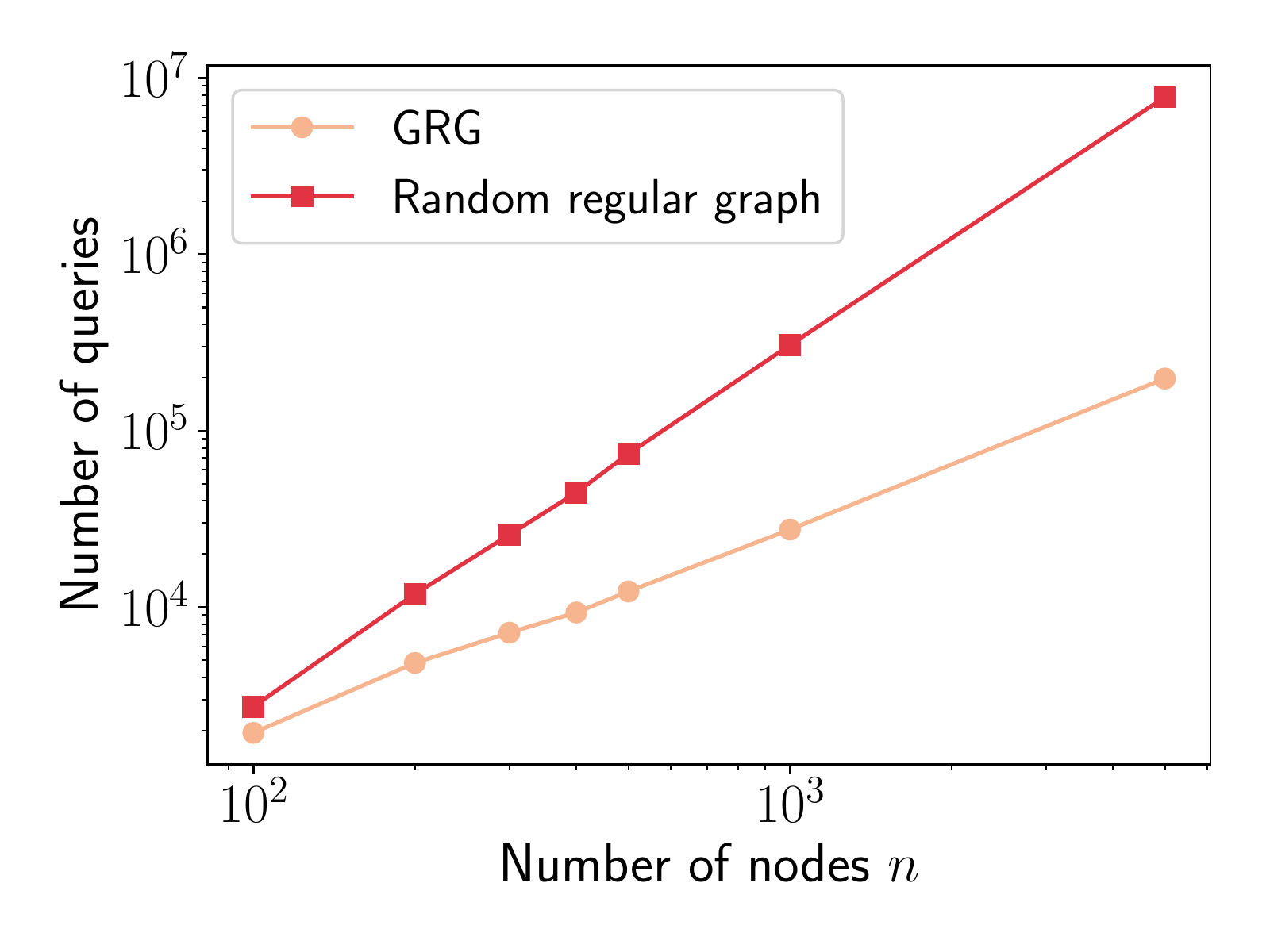}
        \caption{$r = n^{0.3}$.}
        \label{sfig:dense_comparison}
    \end{subfigure}
    \caption{Simulated query complexity of the \simple algorithm for a GRG with a random regular graph with similar average degree, in a sparse and dense setting.} \label{fig:RRG}
\end{figure}

\begin{figure}[tb!]
    \begin{subfigure}{0.45\textwidth}
        \includegraphics[width = \textwidth]{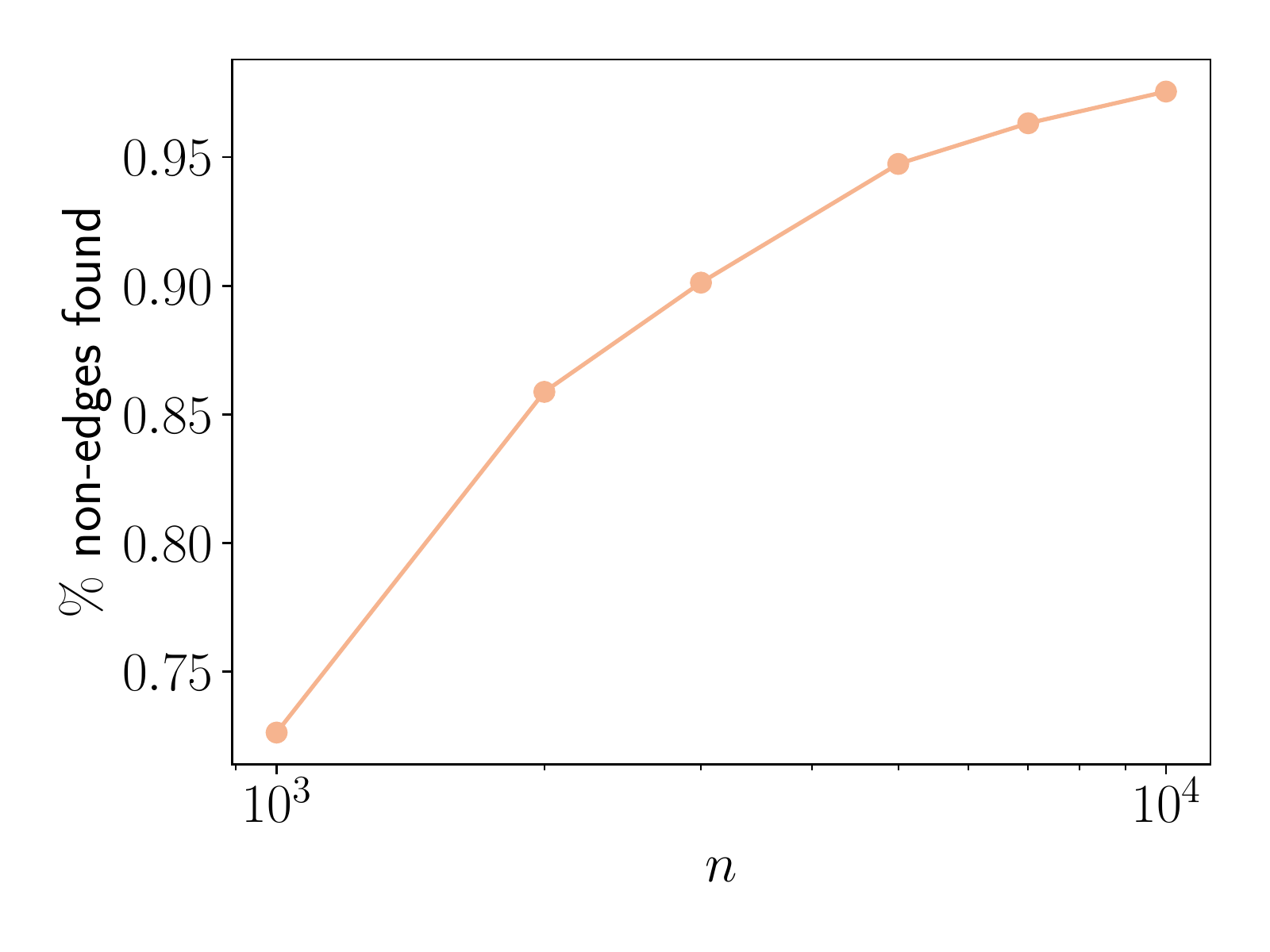}
        \caption{Sparse GRG, $r=2\sqrt{\log(n)}$.}
        \label{sfig:sparse_non_edges}
    \end{subfigure}\hfill
       \begin{subfigure}{ 0.45\textwidth}
        \includegraphics[width = \textwidth]{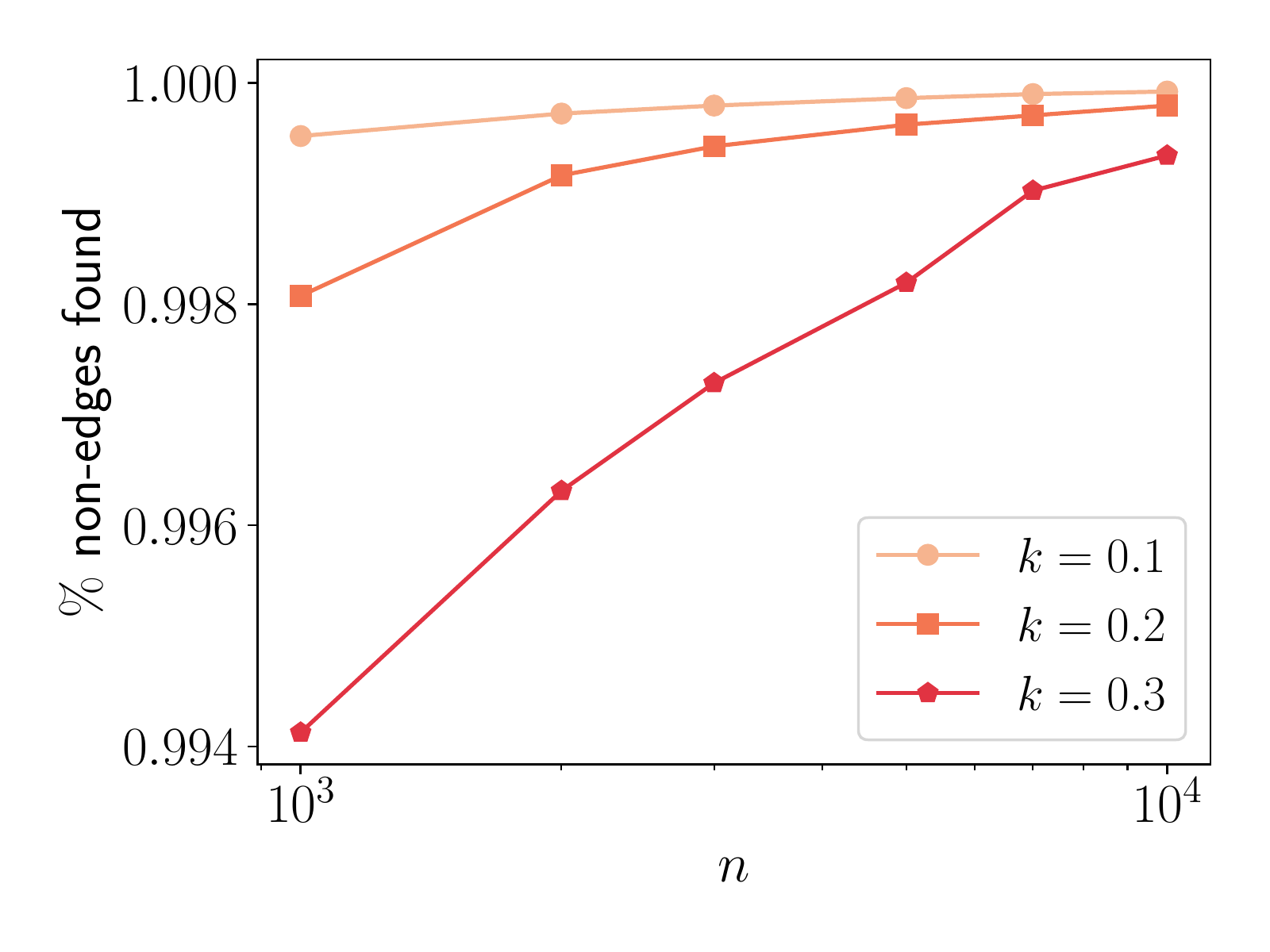}
        \caption{Dense GRG, $r = n^k$.}
        \label{sfig:dense_non_edges}
    \end{subfigure}
    \caption{Percentage of discovered non-edges after the first step in the \simple algorithm.} \label{fig:non_edges}
\end{figure}

Finally, we confirm the results of Lemma \ref{lem:connection_to_theorem} in Figure \ref{fig:non_edges} and show that with the given seed set $S$, where $|S| = n^{\epsilon}\log(n)$, $\epsilon = k$, querying the nodes with all seeds already results in identifying at least $75 \%$ of the non-edges, for both dense and sparse graphs. In particular, the simulations show that for large $n$, approximately $95\%$ of the non-edges is identified.  For the sparse regime, we chose $|S| = 4$. Figure \ref{sfig:sparse_non_edges} shows that in practice, almost all non-edges are found after the first step in the \simple algorithm. The difference between our theoretical bound of $75\%$ and the simulation results is mainly due to the fact that the lower and upper bound as given in Theorem \ref{thm:lowerbound_upperbound} are not tight.

\section{Conclusion}
In this paper, we have shown that the \simple algorithm can provide a fast reconstruction of dense geometric random graphs, whereas until now only results were known for sparse bounded-degree graphs. Moreover, we prove that the \simple algorithm is almost optimal for the class of geometric random graphs, as the query complexity asymptotically matches the number of edges.\\

We have shown that the underlying geometry of the GRG helps in fast reconstruction, due to the local character of edges (i.e., nodes are only connected if they are close). Due to the local edge property, we have shown that querying only 4 seed nodes with all nodes in the graph already results in at least $75\%$ of the non-edges detected. Therefore, we conjecture that the \simple algorithm also efficiently reconstructs other types of geometric random graphs, such as the hyperbolic random graph, which is possibly an even better representation of the Internet topology \cite{boguna2010sustaining}, or the geometric inhomogeneous random graph \cite{bringmann2019geometric}.\\

Another future research direction could be to design algorithms to \emph{approximately} recover the graph, for example by terminating the \simple algorithm after a given number of queries, or only query a selection of nodes.
For the verification problem, it would be interesting to investigate whether it is possible to obtain a lower query complexity for approximate recovery.

\bibliographystyle{siam}
% \bibliography{biblio}

\end{document}